\documentclass[runningheads]{llncs}
\usepackage{graphicx}
\usepackage[vlined]{algorithm2e}
\SetKwComment{Comment}{/* }{ */}
\RestyleAlgo{ruled}
\DontPrintSemicolon 
\SetKw{KwInSet}{in}
\usepackage{lmodern}
\usepackage{amsfonts}
\usepackage{amsmath}

\usepackage{tikz}
\usetikzlibrary{arrows,shapes,snakes,petri,topaths,automata}
\usepackage{tkz-berge}
\usepackage[position=top]{subfig}

\usepackage{pgf}

\usepackage{graphicx}
\graphicspath{ {graphics/} }

\begin{document}
\title{Locating Robber with Cop Strategy Graph: Subdivision vs. Multiple Cop}
%
%
\author{Shiqi Pan}
\authorrunning{S. Pan}
%

\institute{\email{panshiqi.psq@gmail.com}}
\maketitle              
\begin{abstract}
We consider the Robber Locating Game, where an invisible moving robber tries to evade the pursuit of one or more "helicopter" cops, who send distance probes from anywhere on the graph. In this paper, we attempt to propose two useful constructions for general problems in this game: a \textit{state} variable that describes the available game information for the cops, and a \textit{Cop Strategy Graph} construction that presents all possibilities of the game given a deterministic cop strategy. Then we will use them, along with algorithms and pseudo-code, to explain the relationship between two graph parameters, the localization number $\zeta$ and the subdivision number $\eta$. Researchers have shown that $\eta=O(\zeta)$ and $\zeta \neq O(\eta)$. We will revisit their proofs, consolidate the essential correspondence between the two numbers via our proposed constructions, and show an explicit result for $\zeta$ in terms of $\eta$, the capture time, and the graph diameter.
\end{abstract}
\section{Introduction}
Cops and Robbers is a widely studied graph pursuit game. Its basic setup includes a finite countable connected graph and two sides of players, robbers and cops. It has many variations. And in this paper, we focus on one game variation called the \textit{Robber Locating Game}. 

In the Robber Locating game, there are one or more "helicopter" cops and a single invisible robber, i.e., the location of the robber is inaccessible to the cops. The game starts with the robber choosing its initial position. Then in each round, the two sides take alternative turns to play: the cops first simultaneously send distance probes from \textit{anywhere} on the graph, each receiving the distance between the probed vertex and the robber's location; then the robber moves to the neighbor vertices or stays unmoved. 
The cops win if they manage to locate the robbers by identifying their position. We are particularly interested in the cops' strategy. A strategy is called \textit{cop-winning} if it secures a win for the cops in finite time. And the graph is called \textit{localizable} if there exists a \textit{cop-winning} strategy, and not \textit{localizable} if not.

Many studies have been done on the \textit{localization number}. Given a graph $G$, the localization number of $G$, denoted $\zeta$, is the minimum number of cops for $G$ to be localizable. Seager studied graphs with one cop~\cite{ref_article1}~\cite{ref_article2} when she first introduced the game, and her research focuses on special graphs such as cycles, complete graphs, and subdivisions of graphs. More studies followed on the topic of subdivision later on. A subdivision of graph $G$ with an integer $m$, denoted $G^{1/m}$, is when each edge of the $G$ is replaced by a path of length $m$. Games on subdivisions  usually assume one cop, as the added paths already put a disadvantage on the robber by slowing it down. We call the minimum value of $m$ such that $G^{1/m}$ is localizable the \textit{subdivision number} of $G$, denoted $\eta$. In Haslegrave, Johnson and Koch's papers~\cite{ref_article4}~\cite{ref_article5}, the bounds on $\eta$ are proved as $n/2$ for $n$-vertex $G$. In their later paper~\cite{ref_article6}, they investigate the relationship between localization numbers $\zeta$ and $\eta$.

Many other graph parameters on various game variations have also been studied.~\cite{ref_article3}~\cite{ref_article7}~\cite{ref_article8} explore a variation with perfect information to all players, and~\cite{ref_article9}~\cite{ref_article10} further incorporate complexity analysis in the research of such variation. Parameter-wise, for example,~\cite{ref_article11} introduced the \textit{capture time}, denoted $capt$, in the Robber Locating Game. It is defined as the minimum number of rounds for the cop to guarantee to win. ~\cite{ref_article12} also studies a similar concept called \textit{escape length} in the game variation \textit{Rabbit and Hunter} where the rabbit can "jump" to any vertex during its round. 

Previous papers have used many tools to prove the localizability of graphs, such as verbal explanations, graphs, and tables. But so far there haven't been many commonly used variables and terms shared across the research. In this paper, we intend to formalize some common terms that many studies have already implicitly relied on, and provide a general and explicit basis for problems in general graph pursuit games. Specifically, we will purpose a \textit{state} variable that describes the game process and a construction called the \textit{Cop strategy Graph}, which is a graph based on the state variable that presents all possibilities of a game given a deterministic cop strategy. In Section 2, we will present formal definitions and the process of building the cop strategy graph.

The cop strategy graph is helpful with illustrating the game process under a certain cop strategy, as well as relating strategies and games. In Section 3, we will use them to explore the relationship between the localization number $\zeta$ and the subdivision number $\eta$. The former will be studied on a Robber Locating Game with one robber and multiple cops; we call it the \textit{multiple-cop game}, denoted $Game_{cop}$. And the latter will be defined on the game with one robber and one cop on a subdivision graph; it is called the \textit{subdivision game} and denoted as $Game_{subs}$. Haslegrave, Johnson, and Koch's showed that $\eta = O(\zeta)$~\cite{ref_article6}. However, for the important observation of the correlation between the two games, their paper didn't provide detailed and explicit explanations. Thus, we will revisit their claim in section~\ref{section:subs} and in particular, establish an explicit correspondence between the games with \textit{state}. Their paper also showed that $\zeta \neq O(\eta)$. In section~\ref{section:cop}, we use the Cop Strategy Graph to again establish a correspondence with the two games. For an explicit relationship, with $capt$ being the minimum number of rounds for cops to secure a win and $\delta$ being the graph diameter, we will show that $\zeta=O(2^{capt/\eta}16^\eta \delta^{2\eta})$.

\section{State and Cop Strategy Graph}
First, we define several variables. For any game, the \textit{robber set} $R_i$ is defined as the set of possible positions of the robber \textit{after} the cops probe and \textit{before} the robber moves in round $i$ and the \textit{extended robber set} $X$ is that \textit{after} the robber moves. $X_0=V(G)$, and $X_{i}=CloseNeighbor(R_i)$ for any $i$, where the \textit{close neighbor} includes the vertices themselves and their adjacent vertices. Cop wins if $|R_i| = 1$.


In this paper, we assume that the cop strategies are \textit{deterministic} in terms of all accessible information in the game. And \textit{state}, denoted $\phi$, is defined as the set of information that may \textit{affect} the cop's probes. It may include the (extended) robber set, round numbers, previous probing results, and so on. The cops take different probes \textit{if and only if} the states are different. Then the cop strategy $A$ is defined as a deterministic function that takes in $\phi$ and outputs the probes for cops to make. Probes can be written as $A(\phi)$.

Now, we construct the \textit{Cop Strategy Graph}. Given a strategy $A$, we denote its cop strategy graph as $H(A)$. Vertex set $V(H)$ contains states as their values. And the graph starts with a node with the initial state. An edge exits from state $\phi_1$ on one level to $\phi_2$ on the next level \textit{if and only if} there exists a valid set of probing results from probes $A(\phi_1)$ that updates the state to $\phi_2$, and the edge's value is the probing results. State $\phi$ is a leaf \textit{if and only if} it is a \textit{terminating state}, i.e., the robber set has of size 1 and the cops win.



See Figure~\ref{figure:copStrategyGraph} for an example. Given a graph $G$, consider the game with one single cop (Graph (2)) and a cop strategy. The cop first probes $C$ by the strategy, locating robbers at $C$, $A$, or $B$ if the probing result is 0, 1, or 2 respectively; if the result is 3, then the robber sets $R=\{D,E\}$, $X=CloseNeighbor(R)\{B,D,E\}$. Then the cop probes $D$, where it is able to locate the robber with any probing results. Similarly. a strategy for the game with two cops is presented in Graph (3).
\begin{figure}
\includegraphics[width=\textwidth]{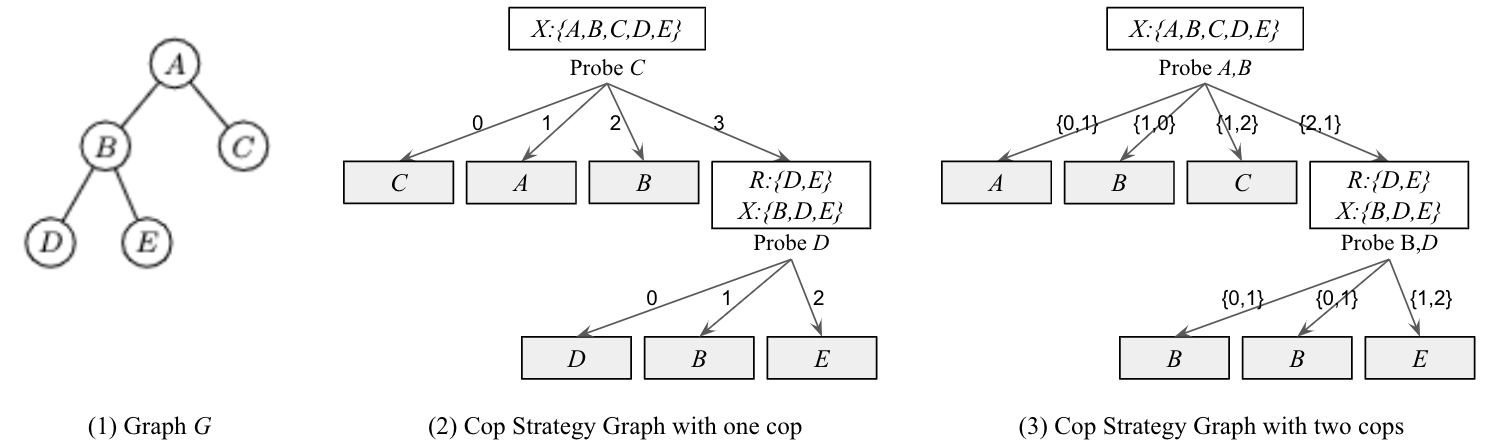}
\centering
\caption{Cop Strategy Graphs}
\label{figure:copStrategyGraph}
\end{figure}

\begin{theorem}
The cop strategy $A$ is cop-winning if and only if its corresponding graph $H_A$ is finite.
\end{theorem}
\begin{proof}
$A$ being \textit{cop-winning} is equivalent to that for any robber's route, the cop can always reduce possible robber locations to one in finite time with $A$, which is equivalent to each branch of $H_A$ terminating on a leaf, i.e., $H_A$ is finite.
\end{proof}

\section{Multiple Cops vs. Subdivision}
Now, we use the state variable and the cop strategy graph to explore the relationship between the localization number and subdivision number.

Set a graph $G$. We study the localization game $Game_{cop}$ and the subdivision game $Game_{subs}$. Note that for ease of reading, the footnotes $_{cop}$(or $_c$) and $_{subs}$(or $_s$) will also be used in other parameters to flag the two games respectively. Previous papers have shown important conclusions regarding the relationship. ~\cite{ref_article6} proves that $\eta = O(\zeta)$ and that a factor of 2 is the best possible. It focuses mainly on the specific actions of the cop in $Game_{subs}$ and puts less emphasis on the relationship between the two games. However, a closer look at the correspondence, including the transformation of probing vertices and results, is necessary. Thus, we will revisit their cop strategy and use the state variable to present a detailed analysis of the relationship and the deduction method in this section, which is the foundation of the validity of the strategy. ~\cite{ref_article6} also proves that $\zeta \neq O(\eta)$, and in particular, $\forall m \geq 3$, there exists $G$ such that $\eta=m, \zeta=2^{m-1}$. In Section~\ref{section:subs}, we use the Cop Strategy Graph to relate the two games and present an explicit result that $\zeta=O(2^{capt/\eta}16^\eta \delta^{2\eta})$.

\subsection{$\eta = O(\zeta)$)}\label{section:subs}
Assume there exists a cop-winning strategy $A_{cop}$ for the multiple-cop game with $\zeta$ cops. We will show that for any subdivision number $m=O(\zeta)$, $Game_{subs}$ is localizable. Our key is to ``mock" the robber in $Game_{cop}$ that makes the \textit{"same"} moves as the ``real" one in $Game_{subs}$ and to get probing "hints" from $A_{cop}$.

Before going into the algorithm, we introduce several concepts for subdivisions of graphs. For details, readers are directed to~\cite{ref_article4}. The paths that are added to $G^{1/m}$ are called \textit{threads}. Their endpoints, i.e., vertices in $G^{1/m}$ that correspond to vertices in the original graph $G$, are called \textit{branch vertices}. We denote the relationship of branch vertex $b$ in $G^{1/m}$ and its corresponding vertex $v$ in $G$ as $b = v^{1/m}$ and $v = b^{m}$, and we say $b$ and $v$ are ``\textit{subdivisionally} equal" or ``\textit{subdivisionally} equivalent". The same notations are also applied to vertex sets in the following sections. When $m$ is odd, the \textit{midpoint} is the one vertex in the middle of every thread; and when $m$ is even, \textit{near-midpoints} are the two in the middle. By probing any branch vertex on the graph and taking the module of the result distance, the cop can tell whether the robber is at a branch vertex, (near)-midpoint, or neither.

The robber remains in one thread until visiting a branch vertex. And it is always closest to one or two vertices, two when the robber is at the midpoint of a thread with an odd $m$. Its closest vertex remains the same between two consecutive visits to midpoints. We say that the robber is within the vicinity of a branch vertex $b$ if $b$ is its nearest branch vertex. We thus observe a natural link between the robber being within the \textit{vicinity} of a branch vertex $b$ in $G^{1/m}$ and it being at $v=b^m$ in $G$. And we will base our strategy on this link.

\subsubsection{Cop Strategy}
We propose a cop strategy $A_{subs}$ in $Game_{subs}$. It only probes branch vertices. And the game is considered as three separate stages in regard to the robber's movement. Stage 1 is before the robber visits any branch vertex; Stage 2 is when the robber moves between branch vertices along threads; and Stage 3 is after the robber last visits a branch vertices until it is caught.

In Stage 1, we probe at random. The robber is guaranteed to remain in one thread, and thus it is located once both endpoints of its current thread have been probed. 

For Stage 2, similar to~\cite{ref_article6}, we make sure that the cop probes the corresponding vertices of probing vertices provided by $A_{cop}$ while it is in the vicinity of a branch vertex, probing to "identify" that branch vertex. Along the robber's route, it may visit multiple branch vertices and midpoints. We define a \textit{stride} as the rounds between the robber's two consecutive visits to branch vertices while passing midpoints. Namely, stride $i$ includes the rounds \textit{exclusively after} the robber moves to a branch vertex, denoted $b_{i-1}$, and \textit{inclusively before} it visits the consecutive one, denoted $b_i$. We also call the rounds in $Game_{cop}$ as strides too as there is a correspondence between the two in this stage.

With $\phi_{c}$ as the initial state, the cops probe in $Game_{subs}$ based on strategy $A_{probs}$ in each stride $i$ as in Algorithm~\ref{alg:strategy1}:
\begin{algorithm}
\caption{Cop Strategy in Subdivision Game}\label{alg:strategy1}
\KwIn{Cop strategy $A_{cop}$, graph $G$}
\KwOut{\textit{True} if cop wins}
\SetKwFunction{FStrategicProbe}{StrategicProbe}
\SetKwFunction{FDeduce}{Deduce}
\SetKwProg{Pn}{Function}{:}{\KwRet}

\Pn{\FStrategicProbe{$P_{cop}$}}{
    $ProbeResults$ = $\{\}$
    \For{$p_{cop} \in P_{cop}$} {
        Probe $p_{cop}^{1/m}$\;
        \If {Robber is at a midpoint} {
            \KwRet $StrategicProbe(P_{cop})$\;
        }
        Add the result to $ProbeResults$\;
    }
    \KwRet $ProbeResults$\;
}

\Pn{\FDeduce{$Result_{subs}$}}{
    $Result_{cop}=\{\}$\;
    \For{$d \in Result_{subs}$}{
        $Result_{cop}$.add($d/m$)\;
    }
    \KwRet $Result_{cop}$\;
}
\;

Initialize $\phi_c=State\{ExtendedRSet: V(G)\}$, $i=1$\;

\While{True}{
    Get $P_{cop} = A_{cop}(\phi_{c})$, the multiple-cop probes at round $i$\;
    \While {robber visits non-midpoint vertices} {
        Probe randomly and \KwRet true if the robber is located\;
    }
    Result = $StrategicProbe(P_{cop})$\;
    \KwRet true if the robber is located\;
    \While {True} {
        Probe randomly and \KwRet true if the robber is located\;
    }
    Get $Deduce(Result)$, input to $Game_{subs}$, and update $\phi_{c}$\;
}
\end{algorithm}

\subsubsection{Deduction of Probing Results}
Consider only the \textit{strategic probes}, the ones based on multiple-cop strategy $A_{cop}$. We now show that the $Deduce$ function in the algorithm, with probing results in $Game_{cops}$ as input and those for $Game_{subs}$ as output, maintains a correspondence of the robber sets in the two games.

\begin{lemma}\label{lemma:deduce1}
$Deduce$ maintains the \textit{subdivisional equality} of the robber sets $R_c$ in $Game_{cops}$ and $R_s$ in $Game_{subs}$ in each stride.
\end{lemma}

\begin{proof}
Let $q$ be a probe in $Game_{cops}$ by $A_{cops}$ in an arbitrary stride $i$. Its counterpart $q=p^{1/m}$ must be probed in stride $i$ in $Game_{subs}$ too. Let $v$ be a robber location that complies with the probing result of $p$ in $Game_{subs}$. Let $v_1, v_2$ be the endpoints of the thread $v$ is on, where $v_1$ is strictly closer to $v$. Let the corresponding vertices in $G$ be $w_1=v_1^{m}, w_2=v_2^m$. We prove below that given $dist_s(p,v)$, $Deduce()$ outputs $dist_s(q,w_1)$, and thus, $w$ complies with the probing result of $q$.

Denote $d_1=dist_s(p,v_1), d_2=dist_s(p,v_2)$ and $x=dist_s(v_1,v), m-x=dist_s(v_2,v)$. Since $v$ is close to $v_1$,  $0<x<\lfloor m \rfloor / 2$.

See Figure~\ref{figure:deduce1} for an illustration.
\begin{figure}
\begin{tikzpicture}[scale=1.5, auto,swap]
\tikzstyle{vertex}=[circle,draw,minimum size=15pt,inner sep=0pt]
\tikzstyle{branch vertex} = [vertex,minimum size=12pt,fill=black!10]
\tikzstyle{edge} = [draw,thick]
\tikzstyle{right edge} = [draw,thick,bend right]
\tikzstyle{length edge} = [draw,thick,]
\tikzstyle{thread edge} = [draw,thick,dashed]
\tikzstyle{weight} = [font=\small]
    \node[branch vertex] (v) at (0.5,0) {$v$};
    \node[vertex] (v1) at (0,0) {$v_1$};
    \node[vertex] (v2) at (1.5,0) {$v_2$};
    \node[vertex] (p) at (1,1) {$p$};
    \path[thread edge] (v1) -- node[weight] {$x$} (v);
    \path[thread edge] (v) -- node[weight] {$m-x$} (v2);
    \path(v2) edge [bend right] node {$d_2$} (p);
    \path(p) edge [bend right] node {$d_1$} (v1);

    \node[vertex] (w1) at (2.5,0) {$w_1$};
    \node[vertex] (w2) at (4,0) {$w_2$};
    \node[vertex] (q) at (3.5,1) {$q$};
    \path(w2) edge [bend right] node {} (q);
    \path(q) edge [bend right] node {} (w1);
    \path(w2) edge [] node {} (w1);

    \node[text width=3cm] at (1.5, -0.45) {$G^{1/m}$};
    \node[text width=3cm] at (4, -0.45) {$G$};
\end{tikzpicture}
\centering
\caption{Probing Result Deduction}
\label{figure:deduce1}
\end{figure}
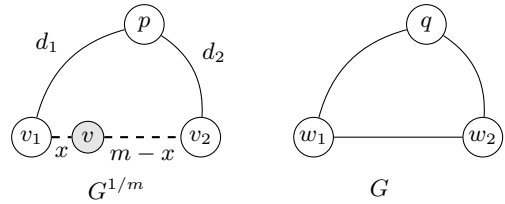

Paths between $p$ and $v$ must go through either one $v_1$ or $v_2$. Thus,
$$
dist_s(p,v)=\begin{cases}
    d_1+x \equiv x \mod m, & \text{if $path(p,v)$ contains $v_1$}\\
    d_2+m-x \equiv m-x\mod m, & \text{if $path(p,v)$ contains $v_2$}
\end{cases}
$$

\textbf{Case 1: $dist_s(p,v) < m/2 \mod{m}$}. This means that $path_s(p,v)$ contains $v_1$, and thus $\lfloor dist_s(p,v)/m \rfloor = \lfloor \frac{d_1+x}{m} \rfloor = \frac{d_1}{m} = dist_p(q,w_1)$.

\textbf{Case 2: $dist_s(p,v) > m/2 \mod{m}$}. This means that $path_s(p,v)$ contains $v_2$, so $d_2+m-x < d_1+x$. Since $d_1,d_2$ are multiples of $m$, and $0<x<\lfloor m \rfloor / 2$, $d_1 \geq d_2 + m$. By Triangle Inequality, the triangle with points $p$, $v_1$, and $v_2$, $d_1 \leq d_2 + m$. Therefore, $d_1 = d_2 + m$. $\lfloor dist_s(p,v)/m\rfloor+1 = \lfloor \frac{(d_2+m-x)+ m}{m} \rfloor = \frac{d_2 + m}{m} = \frac{d_1}{m} = dist_p(q,w_1)$.

Thus, for both cases, the deducted result $round(dist_s(p,v)/m)$ is always equal to $dist_p(q,w_1)$, and thus $w$ complies with the probing result from $q$. Similarly, the other direction is also true: if $w$ is a possible location in regard to the probing result from $q$, vertices in $G^{1/m}$ in the vicinity of $v$ are also possible locations in regard to the result from $p$. And by a simple induction on stride number, we show that the \textit{subdivisional equality} of the robber sets $R_c$ and $R_s$ is maintained.
\end{proof}

\subsubsection{Proof of Correctness}
We now show that the subdivision strategy $A_{subs}$ that we propose based on \textit{cop-winning} multiple-cop strategy $A_{cop}$ is \textit{cop-winning}.

\begin{theorem}
$A_s$ is \textit{cop-winning}.
\end{theorem}
\begin{proof}
As proved above, if we only consider the strategic probes in $Game_{subs}$, the robber sets in the two games would be \textit{subdivisionally} the same. The \textit{non-strategic} random probes in $Game_{subs}$ during Stage 1 or 2 only eliminate elements from the robber set $R_{s}$. Thus, $R_s \subseteq R_c$ anytime in the games.

Now, assume in some point in the games, $|R_{c}|=1$, then $|R_{s}| \leq 1$. Note that such must happen before the game enters Stage 3, and we will be able to locate the robber in the first two Stages. We have thus shown that $A_s$ is \textit{cop-winning}.
\end{proof}

\subsection{$\zeta=O(2^{capt/\eta}16^\eta \delta^{2\eta})$}\label{section:cop}
Now, we look at the relationship between the localization number and the subdivision number in the other direction. Set a graph $G$. Given a \textit{cop-winning} strategy $A_{subs}$ for $Game_{subs}$ on graph $G^{1/\eta}$ , we will construct a \textit{cop-winning} strategy $A_{cop}$ for $Game_{cop}$ with $k$ cops, for any $k \geq 2^{capt/\eta}16^\eta \delta^{2\eta}$.

We again establish a correspondence between the robber sets for the two games, and to do so, we first set some special limits on the robber's behavior in $Game_{subs}$. This is valid as we "mock" the robber in $Game_{subs}$ to solve the target game $Game_{subs}$. First, we limit the robber to only choose branch vertices as its initial position. Next, the robber takes \textit{strides} during the game and their definition is slightly different from the previous. In the current problem, each stride contains exactly $\eta$ rounds, and the robber lands on a branch vertex at the end of each stride. If it moves to another branch vertex in a stride from the last consecutive stride, it can only do so directly without backtracking; if it ends up at the same vertex, it is limited to remaining on the vertex for all $m$ rounds or moving to a midpoint and back. The \textit{round index} for a stride is defined as the index of the round within the $\eta$ rounds in a stride and takes the value of 1 to $\eta$.

We denote the branch vertex the robber visits at the end of each stride $i$ as $b_i$. Its initial position is denoted as $b_0$. Then we construct the cop strategy graph $H$ of the cop strategy $A_{subs}$ for this special setup of the subdivision game. It is rooted at the initial state with the extended robber set as all vertices.

\begin{figure}
    \centering
    \includegraphics[width=\textwidth]{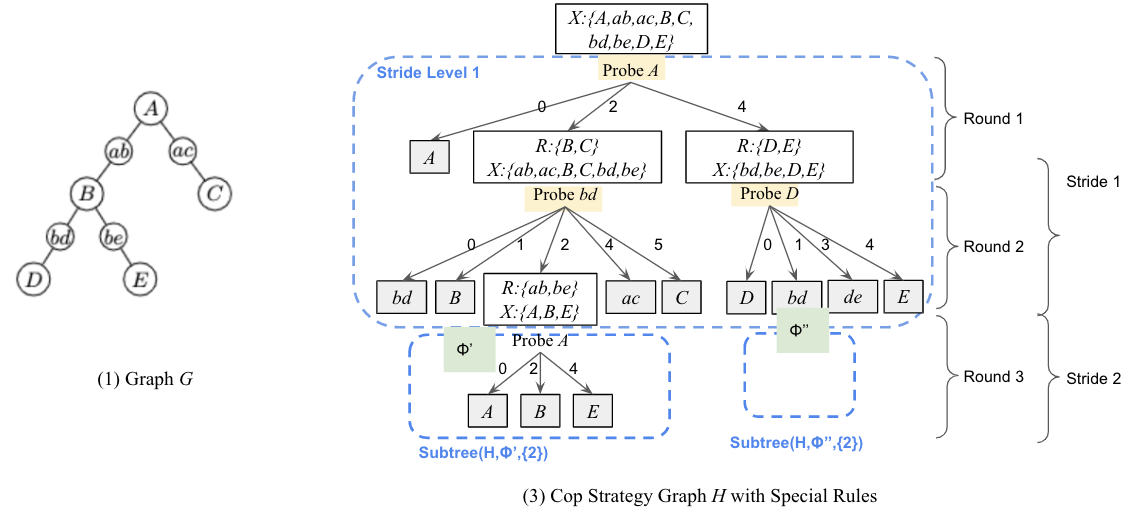}
    \includegraphics[width=\textwidth]{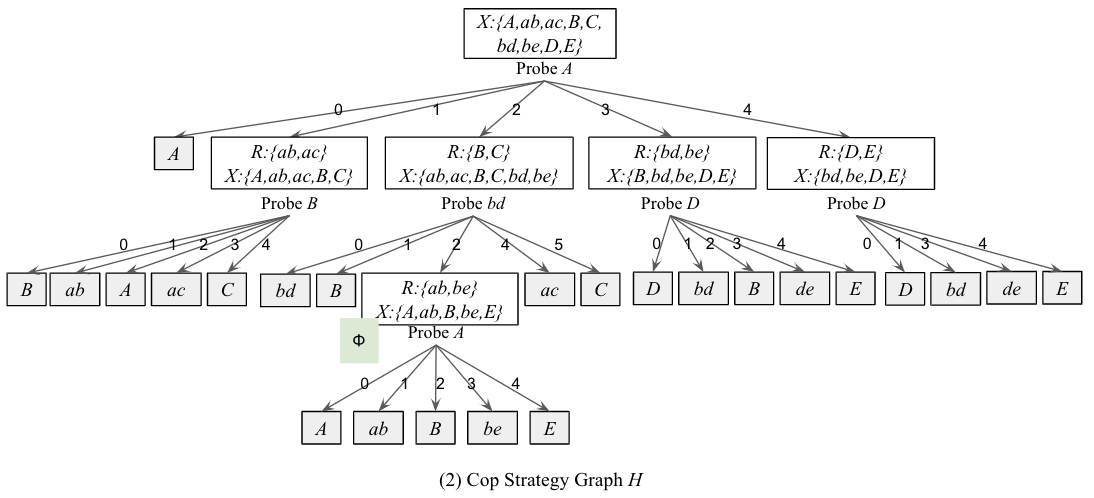}
    \caption{Cop Strategy Graph for Subdivision Games}
    \label{fig:subs}
\end{figure}

As an example, Figure~\ref{fig:subs} presents a graph $G$ (1) and a complete cop strategy graph $H$ (2). With our special setup, $H$ is reduced (3). In specific, we remove 1 the non-branch vertices in the robber sets in round 1, as well as $ab, be$ from the extended robber set in round 2 when transforming $\phi$ to $\phi'$, as they don't comply with the robber's allowed movements.

For $H$ with the rules, we define \textit{stride level}. Stride level $i$ of $H$ as all possible probes during stride $i$. And we denote the sub-tree rooted at a state $\phi$ as $SubTree(H, \phi)$. $SubTree(H, \phi, \{l_1, l_2\})$ represents the set of vertices on stride level $l_1, l_2$ of $H$ that are contained in $SubTree(H, \phi)$. $SubTree(H, \phi, l) \subset SubTree(H, \phi', l)$ where $\phi'$ is an ancestor state of $\phi$. See Graph (3) in Figure~\ref{fig:subs} for illustrations. We also introduce the term \textit{corresponding endpoints} of a vertex $p$ in $G^{1/\eta}$ as the vertices in $G$ that correspond to the two endpoints of the thread that $p$ is on, denoted $CorrEnd(p)$. If $p$ is itself a branch vertex, the corresponding endpoint is $p^m$ itself.

\subsubsection{Cop Strategy}
We now present the cop-winning strategy $A_{cop}$ in $Game_{cop}$ on $G$. Its main difference from the one in Section~\ref{section:subs} is that it creates an \textit{one-to-multiple} relationship of states in the two games.

During $Game_{subs}$, we maintain a set of states $\Phi_{s}$, which contains all possible states in the subdivision game. It is initialized with a single initial state whose extended robber set is all vertices in $G^{1/\eta}$. At the start of the game, we probe in $G$ the corresponding endpoints of $SubTree(H, \phi_{s}, \{1\})$. Starting from the next round (counted as round 1 for convenience), $A_{cop}$ acts according to Algorithm~\ref{alg:strategy} until the cops win. The probes starting for round 1 are regarded as the \textit{strategic probes}. And to emphasize the one-to-one correspondence of rounds and strides in the two games, we also call \textit{rounds} in $Game_{cop}$ as \textit{strides}.

\begin{algorithm}
\caption{Cop Strategy in Multiple-cop Game}\label{alg:strategy}
\KwIn{Cop strategy graph(of subdivision game) $H$, graph $G$}
\KwOut{\textit{True} if cop wins}

Initialize $\phi_s=State\{ExtendedRSet: V(G^{1/\eta})\}$, $\Phi_s=\{\phi_s\}$, $i=1$\;

\While{True}{
    Probe all corresponding endpoints of $\bigcup_{\phi_s \in \Phi_s} SubTree(H,\phi_s,\{i,i+1\})$, and if there exist terminating states, also probe those of the terminating robber's locations; store the results in $D_{all}$\;
    \KwRet True if the robber is located\;
    
    \While{$\Phi_s$ not empty}{
        $\phi$ = $\Phi_s$.pop()\;
        Get from $D_{all}$ $D=SubTree(H,\phi,\{i\})$ with the last and current $\phi$\;
        Deduce probing results and get updated states for $Game_{subs}$ via $DeduceAndUpdate$ function; append updated states to $\Phi_s$\;
    }
    i=i+1\;
}
\end{algorithm}

If $\bigcup_{\phi_s \in \Phi_s} SubTree(H,\phi_s,\{i,i+1\})$ for some stride $i$ contains terminating states in $Game_{subs}$, our probes in $Game_{cop}$ of the corresponding endpoints of terminating robber's locations will determine if they are terminating states in $Game_{cops}$ as well. if so the cop wins, and if not, the state is eliminated from focus. Thus, for the following proofs, we assume that the subtrees don't have terminating states.

\subsubsection{Deduction of Probing Results}
We now take a closer look at the $DeduceAndUpdate$ process. We will show that it maintains the correspondence between the robber set after updating with respective probing results.

See Algorithm~\ref{alg:deduce2} for $DeduceAndUpdate()$. For each stride, it takes in three variables: a graph $A$, the current state in $Game_{subs}$, and a set of probing results $D$, which is a 4-set of probing results in $G$ from each corresponding endpoint of the probe in the previous and current rounds. Different from the previous \textit{Deduce} function defined in Algorithm~\ref{alg:strategy1}, it then outputs \textit{one or two} \textit{sets} of updated states, instead of one.

 \begin{algorithm}
\SetKwFunction{FDeduceAndUpdate}{DeduceAndUpdate}
\SetKwFunction{FDeduceDiffStride}{DeduceDiffStride}
\SetKwFunction{FDeduceSamestride}{DeduceSamestride}
\SetKwFunction{FDeduceRoundMoves}{DeduceRoundMoves}
\SetKwFunction{FDeduceRoundStays}{DeduceRoundStays}

\SetKwProg{Pn}{Function}{:}{\KwRet}
\caption{Probe Result Deduction: Multiple-cops to Subdivision}\label{alg:deduce2}

\Pn{\FDeduceAndUpdate{$D$, A, $\phi$}}{
    UA, UB, VA, VB = D\;
    Set $diff$ = $(UA \neq UB || UA \neq VA)$ \;
    \uIf{diff}{\KwRet \{DeduceDiffStride()\}} 
    \Else{\KwRet \{DeduceSamestride\}}
}
\;
\Pn{\FDeduceDiffStride{$D$, A, $\phi$ }}{
\For{$\phi \in \Phi$}{
    \For{$j \in [k]$}{
        Result = DeduceRoundMoves(D, j,A ($\phi$)) \tcp*{for each round, deduce the result according to the probe vertex $A(\phi)$}
        Input Result to $Game_{subs}$ for round $j$, and update $\phi$ \;
        
    }
    Update $\phi$ in $\Phi$ if not exists\;
}
\KwRet $\Phi$
}
\;
\Pn{\FDeduceSamestride{$D$, A, $\phi$ }}{
\For{$\phi \in \Phi$}{
    \tcc{robber moves}
    $\phi_1=\phi$ \;
    \For{$j \in [k]$}{
        Result = DeduceRoundMoves(D, j,A($\phi_1$))\;
        Input $Result$ to $Game_{subs}$ for round $j$, and update $\phi_1$\;
    }
    \;
    \tcc{robber stays}
    $\phi_2=\phi$ \;
    \For{$j \in [k]$}{
        Result = DeduceRoundStays(D, j,A($\phi_2$))\;
        Input $Result$ to $Game_{subs}$ for round $j$, and update $\phi_2$\;
    }
    Replace $\phi$ with $\phi_1, \phi_2$ in $\Phi$ if not exist\;
}
\KwRet $\Phi$\;
}
\;
\Pn{\FDeduceRoundMoves{$D$, $j$, A, $p$}}{
    Get (uA, uB, vA, vB) the 4 results from probe $p$ from all results $D$\;
    \KwRet $\min(j+m*dist(u,A) +dist(A',p), (m-j)+m*dist(v,A)+dist(A',p), j+m*dist(u,B) +dist(B',p), (m-j)+m*dist(u,B) +dist(B',p))$\;
}
\;
\Pn{\FDeduceRoundStays{$D$, $j$, A, $p$}}{
    Get (uA, uB, vA, vB) the 4 results from probe $p$ from all results $D$\;
    \KwRet $\min(m*dist(u,A) +dist(A,p), m*dist(u,B) +dist(B,p))$\;
}
\end{algorithm}

\begin{lemma}
\textit{DeduceAndUpdate()} maintains the \textit{subdivisional} equality between the robber sets in the two games during each round/stride.
\end{lemma}

\begin{proof}
At each stride $i$ in $Game_{subs}$, given a probe $p$ in $G^{1/\eta}$ in $SubTree(H, \phi, i)$, let $p$ be on a thread with endpoints $a',b'$, and let $a, b$ be their correspondence in $G$. By our cop strategy, both $a$ and $b$ are probed in rounds $i-1$ and $i$, i.e., before and after the robber's turn in round $i-1$. Say that the robber moves from $u$ to $v$ in round $i-1$, then in $D$ we can obtain $dist(u,a), dist(u,b), dist(v,a), dist(v,b)$, for all $p$.
\begin{figure}
\centering
\tikzstyle{vertex}=[circle,draw,minimum size=15pt,inner sep=0pt]
\tikzstyle{sub vertex} = [vertex,minimum size=10pt,fill=black!10]
\tikzstyle{edge} = [draw,thick]
\tikzstyle{right edge} = [draw,thick,bend right]
\tikzstyle{arrow edge} = [draw,thick,arrows]
\tikzstyle{thread edge} = [draw,thick, dashed]
\tikzstyle{weight} = [font=\small]
\begin{tikzpicture}[scale=1.5, auto,swap]
    \node[vertex] (u') at (0,0) {$u'$};
    \node[vertex] (v') at (1.2,0) {$v'$};
    \node[vertex] (a') at (0.2,1.3) {$a'$};
    \node[vertex] (b') at (1.4,1.3) {$b'$};
    \node[sub vertex] (r) at (0.6,0) {$r$};
    \node[sub vertex] (p) at (0.8,1.3) {$p$};
    \path(u') edge [bend left] node {} (a');
    \path(u') edge [bend left] node {} (b');
    \path(v') edge [bend right] node {} (a');
    \path(v') edge [bend right] node {} (b');
    \path[thread edge] (a') -- node[] {} (p);
    \path[thread edge] (p) -- node[] {} (b');
    \path[thread edge] (u') -- node[weight] {$j$} (r);
    \path[thread edge] (r) -- node[weight] {$m-j$} (v');
    \node[text width=3cm] at (1.4, -0.5) {(1) $u \neq v$};
\end{tikzpicture}
\begin{tikzpicture}[scale=1.5, auto,swap]
    \node[vertex] (u') at (0.6,0) {$u'$};
    \node[vertex] (a') at (0.2,1.3) {$a'$};
    \node[vertex] (b') at (1.4,1.3) {$b'$};
    \node[sub vertex] (p) at (0.8,1.3) {$p$};
    \path(u') edge [bend left] node {} (a');
    \path(u') edge [bend right] node {} (b');
    \path[thread edge] (a') -- node[] {} (p);
    \path[thread edge] (p) -- node[] {} (b');
    \node[text width=3cm] at (1.4, -0.5) {(2) $u = v$};
\end{tikzpicture}
\caption{$G^{1/\eta}$ when $u,v$ in $G$ are different or the same}
\label{figure:cop}
\end{figure}
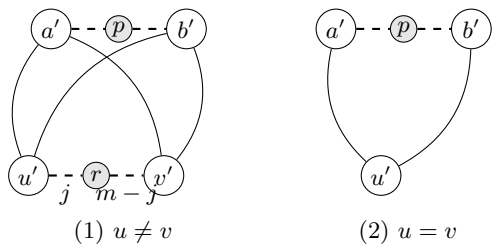

\textbf{Case 1}: If at least of one the probing results are different in the two rounds, i.e., there exists $p$ such that $dist(u,a)\neq dist(v,b)$ or $dist(u,B)\neq dist(v,B)$, then $v, u$ must be different vertices. See the illustration as Figure~\ref{figure:cop} (1). In $G^{1/\eta}$, given a round count $j$, we want to "mock" the robber at the $j$-th vertex from $u'=u^{1/m}$ to $v'=v^{1/m}$. Denote the robber's location as $r$. Then the probing result of $p$ is the minimum length of paths from $r$ and $p$, and we simply return the minimum length of the 4 paths. This is described in the function $DeduceRoundMoves()$.

\textbf{Case 2}: If the probing results in the two rounds are exactly the same, i.e, $\forall p, dist(u,A)=dist(v,A)$ and $dist(u,A)=dist(v,A)$, then $v, u$ are either the same or adjacent with the same distance to the probes. In this case, with one set of probing results in $Game_{cop}$, two sets in $Game_{subs}$ are considered. The first results "mock" the robber remaining at the same vertex, ss Figure~\ref{figure:cop} (2) shows. It is described in the function $DeduceRoundStays()$. The second results, on the other hand, "mock" the robber moving from $u$ towards $v$ with the same probing result. Figure~\ref{figure:cop} (1) presents an illustration, and it is described in the function $DeduceRoundMoves()$.

In both cases, the updated robber sets in the two games contain \textit{exactly} the close neighbors of their previous sets that comply with the probing results. And thus, the correspondence is maintained.
\end{proof}

\subsubsection{Proof of Correctness}
Now, we show that the proposed cop strategy is efficient to win the multiple-cop game. First, we assume that we have an infinite number of cops in $Game_{cop}$, i.e., \textit{all} probing attempts are fulfilled each round, and based on this, we prove the correctness of the strategy. Next, we will calculate the actual cop number required to complete the probes, and provide the bound on localization number $\zeta$ in regard to the subdivision number $\eta$ formally.

\begin{theorem}
Given $A_{subs}$ is cop-winning, $A_{cop}$ is cop-winning.
\end{theorem}
\begin{proof}
The strategic probes maintain the \textit{subdivisional} equality between the robber sets in the two games. And in $Game_cop$, the additional probes in round "0" at the start may further reduce its robber set, and probing the corresponding endpoints of terminating states in $Game_{subs}$ may locate the robber. Thus, $R_c \subseteq R_s$ anytime in the games.

Assume in some stride $i$ in the games, $|R_{subs}|=1$, then $|R_{c}| \leq 1$. Thus, $A_{cop}$ is \textit{cop-winning}.
\end{proof}

Now, we calculate the localization number $\zeta$ required to send all probe attempts each round.

\begin{theorem}
Let $G$ be a graph with subdivision number $\eta$, and let $capt$ be the number of rounds needed for the cops to win the game on graph $G^{1/\eta}$. Let $\delta$ be the diameter of $G$. We show that $\zeta= O(2^{capt/\eta}16^\eta \delta^{2\eta})$.
\end{theorem}
\begin{proof}
Consider the maximum number of states in $Game_{subs}$  maintained in each stride, i.e., the maximum size of $\Phi_s$. $\Phi_s$ remains the same size when probing results in two consecutive rounds are different in $Game_{cop}$, but doubles when the probing results are the same.  With $capt$ rounds, there are $O(capt/\eta)$ strides, and thus $|\Phi_s|=O(2^{capt/\eta})$.

Then consider the number of probes for each round in $Game_{cop}$. We probe in $G$ the corresponding endpoints of  $\bigcup_{\phi_s \in \Phi_s} SubTree(H,\phi_s,\{i,i+1\})$ each round. A stride level in $H$ contains exactly $\eta$ tree levels, and each tree node has $O(\delta_{G^{1/\eta}})$ degree as there are at most $O(\delta_{G^{1/\eta}})$ possible probing results. Since at each tree node, the round index and probe positions are settled, i.e., the distance between the robber/probes and their nearest branch vertex, each path length in $G$ corresponds to 4 distances, and thus $O(\delta_{G^{1/\eta}})=O(4\delta)$.  Each $SubTree(H,\phi_s,\{i,i+1\})$ has $O(2\eta)$ tree levels, so it has in total $O((4\delta_G)^{2\eta})=O(16^\eta \delta^{2\eta})$ vertices.

Therefore, the minimum number of probes required for each round is $\zeta= O(2^{capt/\eta}16^\eta \delta^{2\eta})$.
\end{proof}

\end{document}